\documentclass[a4paper,11pt,twoside,notitlepage,reqno]{amsart}
\makeatletter
\@namedef{subjclassname@2020}{%
  \textup{2020} Mathematics Subject Classification}
\makeatother
\usepackage{fullpage}
\usepackage{tikz,wasysym}
\usepackage{pgf}
\usetikzlibrary{positioning,automata}
\usepackage[font=small,labelfont=bf]{caption}
\usepackage{bbm}
\usepackage{verbatim}
\usepackage[linktocpage]{hyperref}
\hypersetup{
    colorlinks=true,        
    linkcolor=red,          
    citecolor=blue,         
    filecolor=magenta,      
    urlcolor=cyan           
}

\usepackage{amssymb}
\usepackage{url}
\usepackage{adjustbox}
\usepackage{amsthm}
\usepackage{comment}
\usepackage{amssymb,array}
\usepackage{enumitem}
\usepackage{booktabs}

\usepackage{url}
\usepackage{graphics}
\usepackage{latexsym}
\usepackage{epsf}
\usepackage{amsmath,amssymb,amsbsy,amsthm}
\usepackage[hmargin=1.4in,vmargin=1.4in]{geometry}
\usepackage{url}

\usepackage{tikz}
\usepackage[linktocpage]{hyperref}
\hypersetup{
    colorlinks=true,        
    linkcolor=red,          
    citecolor=red,         
    filecolor=magenta,      
    urlcolor=cyan           
}

\addtocontents{toc}{\protect\setcounter{tocdepth}{1}}

\usepackage{eucal}

\DeclareMathOperator{\proj}{proj}

\theoremstyle{plain}
	\newtheorem{thm}{Theorem}
	
		\numberwithin{thm}{section}
	\newtheorem{lemma}[thm]{Lemma}
	\newtheorem{prop}[thm]{Proposition}
	
	\newtheorem{conj}[thm]{Conjecture}

	\newtheorem*{thm*}{Theorem}
	\newtheorem*{lemma*}{Lemma}
	\newtheorem*{prop*}{Proposition}
	\newtheorem*{cor*}{Corollary}
	\newtheorem*{conj*}{Conjecture}

\theoremstyle{definition}
	
	\newtheorem*{example*}{Example}
	
	\newtheorem{remark}[thm]{Remark}

\begin{document}

\title[Intersections of sparse automatic sets]{Quantitative estimates for the size of an intersection of sparse automatic sets}

\author[S. Albayrak]{Seda Albayrak}
\address{Department of Mathematics and Statistics, University of Calgary, Calgary, AB Canada T2N 1N4}
\email{gulizar.albayrak@ucalgary.ca}
\thanks{The first-named author's postdoctoral appointment at the University of Calgary was partially supported by NSERC grant RGPIN-2018-03770 and CRC tier-2 research stipend 950-231716.} 
\author[J. P. Bell]{Jason P. Bell}
\address{Department of Pure Mathematics, University of Waterloo, Waterloo, ON Canada N2L 3G1}
\email{jpbell@uwaterloo.ca}
\thanks{The second-named author was supported by NSERC grant RGPIN-2022-02951.}

\begin{abstract} A theorem of Cobham says that if $k$ and $\ell$ are two multiplicatively independent natural numbers then a subset of the natural numbers that is both $k$- and $\ell$-automatic is eventually periodic.  A multidimensional extension was later given by Semenov.  In this paper, we give a quantitative version of the Cobham-Semenov theorem for sparse automatic sets, showing that the intersection of a sparse $k$-automatic subset of $\mathbb{N}^d$ and a sparse $\ell$-automatic subset of $\mathbb{N}^d$ is finite with size that can be explicitly bounded in terms of data from the automata that accept these sets.
\end{abstract}

\maketitle

\section{Introduction}

Let $k$ be a natural number that is greater than or equal to $2$.  A subset $S$ of $\mathbb{N}$ is $k$-\emph{automatic} 
if there is a deterministic finite automaton with input alphabet $
\Sigma_k=\{0,1,\ldots ,k-1\}$ with the property that the words 
over the alphabet $\Sigma_k$ which are accepted by the 
automaton (where we read words right-to-left) are precisely the words that are base-$k$ expansions of elements of $S$.  One can naturally extend this notion of automaticity to subsets of $\mathbb{N}^d$ with $d\ge 1$, by now working with the input alphabet $(\Sigma_k)^d$.  Then, given a $d$-tuple $(n_1,\ldots ,n_d)$ of natural numbers---after possibly padding some words with $0$ at the beginning---we see there exist words $w_1,\ldots ,w_d$ of the same length with the additional property that $w_i$ is a base-$k$ expansion of $n_i$ for $i=1,\ldots ,d$ (base-$k$ expansions are unique up to some number of leading zeros) and where at least one $w_i$ has no leading zeros.  Then a subset of $\mathbb{N}^d$ is $k$-automatic if there is a finite-state machine with input alphabet $(\Sigma_k)^d$ that accepts precisely the words $(w_1,\ldots ,w_d)$ corresponding to $d$-tuples of natural numbers in $S$.  We refer the reader to the book of Allouche and Shallit \cite{AS} for further background on automata and automatic sets, and we assume that the reader has some familiarity with deterministic finite-state automata. 

As an example, observe that the deterministic finite-state automaton in Figure 1 with input alphabet $\Sigma_2 = \{0,1\}$ accepts the set of words corresponding to binary expansions of elements of the $2$-automatic set $\{3\cdot 2^n+1\colon n\ge 1\}$, where we adopt the usual convention of using doubly circled states to denote accepting states of a finite-state automaton.

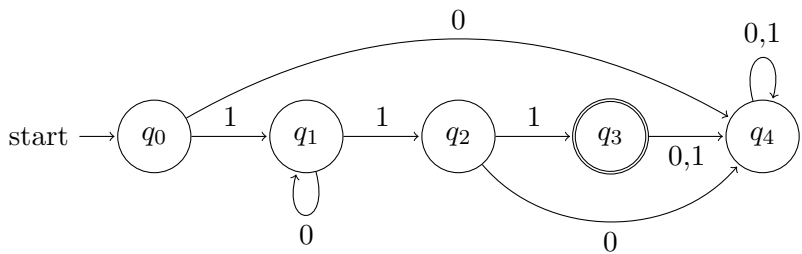
\begin{figure}[!htbp]
\begin{tikzpicture}[shorten >=1pt,node distance=2cm,on grid]
  \node[state,initial]   (q_0)                {$q_0$};
  \node[state]           (q_1) [right=of q_0] {$q_1$};
   \node[state]           (q_2) [right=of q_1] {$q_2$};
   \node[state, accepting]     (q_3) [right=of q_2] {$q_3$};
   \node[state]     (q_4) [right=of q_3] {$q_4$};
  \path[->] (q_1)		(q_1) edge [loop below]   node         {0} ()
   (q_0)		 	edge     	node [above] {1} (q_1)
   (q_0)		 	edge  [bend left]   	node [above] {0} (q_4)
		(q_1)		edge node [above] {1} (q_2)
		(q_2)		edge [bend right =50] node [below] {0} (q_4)
		(q_2)		edge node [above] {1} (q_3)
		(q_3)		edge  node [below] {0,1} (q_4)
		(q_4)		(q_4) edge [loop above]   node         {0,1} ();
  \end{tikzpicture}
  \caption{The finite-state machine generating the set $\{3\cdot 2^n+1\colon n\ge 1\}$.}
\end{figure}

A celebrated result of Cobham \cite{Cobham} shows that if $k$ and $\ell$ are two multiplicatively independent natural numbers greater than one (i.e., there are no solutions to the equation $k^a=\ell^b$ with nonzero integers $a$ and $b$) and $S \subseteq \mathbb{N}$ is a set that is both $k$- and $\ell$-automatic then $S$ is in fact eventually periodic; i.e., there is some fixed positive integer $c$ such that for sufficiently large $n \in \mathbb{N}$,  $n\in S$ implies $n+c\in S$.  A multidimensional version of Cobham's theorem was later given by Semenov \cite{Sem} (see also \cite{MV96}), who showed that a subset of $\mathbb{N}^d$ that is both $k$- and $\ell$-automatic, with $k$ and $\ell$ multiplicatively independent, is a semilinear set (equivalently, a set that is definable in Presburger arithmetic or a set that is automatic with respect to all positive integer bases).

In recent years there have been new proofs and extensions of Cobham's theorem to other settings \cite{adambell2, adambell, bellreg, Durandcobham, durandsub, HierSch, Krebs, SchSin} (see also the survey chapter by Durand and Rigo \cite{durandrigo}).  One particularly interesting extension is recent work of Hieronymi and Schulz \cite{HierSch}, which shows that if one takes Presburger arithmetic and adds a $k$-automatic predicate $X$ and an $\ell$-automatic predicate $Y$, with $k$ and $\ell$ multiplicatively independent, then the resulting structure has an undecidable first-order theory unless one of the two sets is already Presburger definable.  Taking $X$ equal to $Y$, one immediately deduces Cobham's theorem.  In light of this work, it is a natural question to look at the intersection of a $k$-automatic set and an $\ell$-automatic set and to ask to what extent the intersection can be described.   

In general, this question is intractable and many Diophantine questions that lie beyond the scope of currently available methods in number theory can be encoded within this framework.  For example, Erd\H{o}s \cite[p. 67]{Erdos} famously conjectured that the set of powers of two (which is $2$-automatic) and the $3$-automatic set consisting of numbers whose ternary expansions omit $2$ has finite intersection, saying ``\emph{as far as I can see, there is no method at our disposal to attack this conjecture.}” 

Within the theory of automatic sets, however, there is a well-known dichotomy: if $S$ is an automatic subset of the natural numbers, then either there is some natural number $d$ such that $S$ has ${O}((\log n)^d)$ elements of size at most $n$ or there is a positive number $\alpha$ such that $S$ has at least $n^{\alpha}$ elements of size at most $n$ for all sufficiently large $n$ (see, for example, \cite[\S2.3]{Gawrychowski&Krieger&Rampersad&Shallit:2010} or \cite[Proposition 7.1]{BM}).  An automatic set for which the polylogarithmic bound holds is called \emph{sparse}, and this notion again naturally extends to the multidimensional setting.  Sparse automatic sets have arisen naturally in many unrelated contexts \cite{AB, BGM, Derksen, Kedlaya1, Kedlaya2, MoosaScanlon} and form an important subclass of the more general collection of automatic sets.

As an example, the set constructed in Figure 1 is a sparse $2$-automatic set, as there are $O(n)$ elements of size less than $2^n$.  
We refer the reader to \S\ref{sec:sparse} for more background on sparse sets. 

In this paper, we restrict our focus to the problem of giving a description of the intersection of two sparse automatic subsets of $\mathbb{N}^d$ that are automatic with respect to two multiplicatively independent bases.  This setting---while more restrictive than the general setting in which one studies the possible forms an intersection of two automatic sets can take---still captures many interesting number theoretic questions.  Notably, Catalan's conjecture (now a theorem of Mihăilescu \cite{Mih}) asserts that the intersection of the sparse $2$-automatic set $\{2^n+1\colon n\ge 0\}$ with the sparse $3$-automatic set $\{3^m\colon m\in \mathbb{N}\}$ consists only of the numbers $3$ and $9$.  We give the following general finiteness result.
\begin{thm} \label{thm:main1}
Let $k$ and $\ell$ be multiplicatively independent natural numbers greater than or equal to $2$ and let $d$ be a positive integer. If $X$ is a sparse $k$-automatic subset of $\mathbb{N}^d$ and $Y$ is a sparse $\ell$-automatic set of $\mathbb{N}^d$, then $X\cap Y$ is finite and there is an effectively computable upper bound for the size of the intersection in terms of $d, k, \ell$ and data from the minimal automata that accept these sets.
\end{thm}
As in the work of Hieronymi and Schulz \cite{HierSch}, if we take $X=Y$, we see that a subset of $\mathbb{N}^d$ that is a sparse automatic set with respect to two multiplicatively independent bases is necessarily finite, and so in this sense one can view our main result as a quantitative extension of the sparse case of the Cobham-Semenov theorem.

For Theorem \ref{thm:main1}, we in fact give a closed form for upper bounds in terms of just $d$, $k$, $\ell$, and the number of states in the minimal automata accepting $X$ and $Y$ (see Theorem \ref{thm:main2} for explicit bounds; we note that our bounds are not optimized but are rather expressed in a clean form).  One might ask whether one can decide whether the intersection is empty or even whether one can effectively determine the intersection.  Both of these problems are apparently very difficult and connected to highly non-trivial Diophantine questions that are not known to be decidable at this time.

The outline of this paper is as follows. In \S\ref{sec:sparse}, we give a brief overview of sparse languages and sparse sets. In \S\ref{sunit}, we give a brief overview of $S$-unit theory and state the key result we will be using in proving our main theorem. In \S\ref{cobham} we prove a precise version of Theorem \ref{thm:main1}.  Finally, we pose a general conjecture about the form of intersections of sparse $k$-automatic sets with zero-density $\ell$-automatic sets in \S\ref{conjecture}.

\subsection{Notation}
Throughout this paper, given an alphabet $\Sigma$, we let $\Sigma^*$ denote the free monoid consisting of all finite words over the alphabet $\Sigma$.  When $\Sigma=\{u\}$ is a singleton, we write $u^*$ rather than $\{u\}^*$ for $\Sigma^*$.  For an integer $k\ge 2$, we take $\Sigma_k=\{0,\ldots ,k-1\}$, and we let $$[\, \cdot \,]_k: \left(\Sigma_k^d\right)^* \to \mathbb{N}^d$$ denote the map that takes a $d$-tuple of words (here the value of $d$ depends on the context and for much of the paper we take $d=1$) and outputs the $d$-tuple of natural numbers formed by taking the base-$k$ expansions of these words.  So, for example, $[(2110,0020)]_3 = (66,6)$.  In general, we assume that at least one of the words in our $d$-tuple has no leading zeros so that $d$-tuples of natural numbers have unique base-$k$ expansions.   

We also make use of deterministic finite automata (DFAs) throughout this paper. We represent a DFA using a $5$-tuple $(Q,\Sigma,\delta,q_0,F)$, where $Q$ is a non-empty finite set of states, $\Sigma$ is a finite input alphabet, $\delta: Q\times \Sigma \to Q$ is the transition function, $q_0\in Q$ is the initial state, and $F\subseteq Q$ is the set of accepting states.  We note that $\delta$ can be extended inductively as a map from $Q\times \Sigma^*\to Q$ by declaring that $\delta(q,xw) = \delta(\delta(q,x),w)$ for all $w\in \Sigma^*$ and $x\in \Sigma$.
  
\section{Sparse automatic sets}
\label{sec:sparse}
In this section we give a brief summary of sparse languages and sparse sets. Sparse automatic sets and related concepts have been studied by many authors (see, for example, \cite{Gawrychowski&Krieger&Rampersad&Shallit:2010} and references therein).

Given a finite alphabet $\Sigma$ and a language $\mathcal{L}\subseteq \Sigma^*$ over $\Sigma$, we have an associated counting function $$f_{\mathcal{L}}(n):=  \left|\{ w \in \mathcal{L} \colon {\rm length}(w) \le n \}\right|.$$ A regular language $\mathcal{L}$ is \emph{sparse} if $f_{\mathcal{L}}(n) = O(n^d)$ for some natural number $d$.

There is a precise characterization of sparse regular languages, which has been obtained by several authors and is recorded in \cite[Proposition 7.1]{BM}. 
\begin{prop} \label{sparse}
Let $\mathcal{L}$ be a regular language. The following are equivalent: \begin{enumerate}
\item $\mathcal{L}$ is sparse. 
\item $\mathcal{L}$ is a finite union of languages of the form $v_1 w_1^* v_2 w_2^* \dots v_s w_s^* v_{s+1}$, where $s \geq 0$, the $v_i$ are possibly trivial words, and the $w_i$ are non-trivial words over the alphabet $\{0,1,\ldots ,k-1\}$. \label{rem:simple}
\item If $\Gamma = (Q,\Sigma,\delta,q_0,F)$ is a minimal finite automaton accepting $\mathcal{L}$. 
Then $\Gamma$ satisfies the following.
\begin{itemize}
\item[($*$)] If $q$ is a state such that $\delta(q,v)\in F$ for some word $v$ then there is at most one non-trivial word $w$ with the property that $\delta(q,w)=q$ and $\delta(q,w')\neq q$ for every non-trivial proper prefix $w'$ of $w$.
\end{itemize}
\end{enumerate}
\end{prop}
A $k$-automatic subset $S\subseteq \mathbb{N}^d$ is then said to be \emph{sparse} if the sublanguage of $(\Sigma_k^d)^*$ corresponding to base-$k$ expansions of elements of $S$ is a sparse regular language.  Translating Proposition \ref{sparse} into the framework of automatic sets, we see that 
a $k$-automatic subset $S \subseteq \mathbb{N}^d$ is sparse if 
\begin{equation}
\label{eq:piS}
\pi_S(x) = \left|\{(n_1,\ldots ,n_d)\in S \colon n_1+n_2+\cdots +n_d \le x\}\right|={ O}((\log\, x)^d)
\end{equation} as $x$ tends to infinity.  We note that if $S$ is not sparse, then there is some $\alpha>0$ such that $\pi_S(x)>x^{\alpha}$ for $x$ large (cf. \cite[\S2.3]{Gawrychowski&Krieger&Rampersad&Shallit:2010}), and so there is a natural gap separating sparse and non-sparse automatic subsets of $\mathbb{N}^d$.

We will require the following description of special types of sparse sets, from which every sparse automatic subset of the natural numbers can be built by taking finite unions.
\begin{prop} \label{rem:sparse}
Let $k\ge 2$ be a natural number, let $s$ be a nonnegative integer and let $v_0,v_1,\ldots ,v_{s}, w_0,\ldots ,w_s$ be words in $\Sigma_k^*$.  If
$$S=\{[v_0w_1^*v_1 w_2^*\cdots v_{s-1}w_s^*v_s]_k\}$$ then there exist $c_0,\ldots ,c_s\in \mathbb{Q}$ and positive integers $\delta_1,\ldots ,\delta_s$ such that
\begin{equation*} \label{eq: form}
S = \left\{ c_0 + c_1 k^{\delta_s n_s} + c_2 k^{\delta_s n_s + \delta_{s-1} n_{s-1}}+\cdots + c_s k^{\delta_s n_s+\cdots +\delta_1 n_1} \colon n_1,\ldots ,n_s\ge 0\right\}.\end{equation*}

\end{prop}
\begin{proof}
This result is due to Ginsburg and Spanier \cite{GSpan} (see also the proof of \cite[Lemma 3.4]{AB}). 
\end{proof}

\section{Background on $S$-unit equations} \label{sunit}

In this section we give an overview of the theory of $S$-unit equations. Specifically, we require a quantitative version of a result due to Evertse, Schlickewei and Schmidt (see~\cite[Theorem~1.1]{ESS02} and also~\cite[Theorem~6.1.3]{EG}).  We recall that for $z_1,\ldots ,z_n$ in a field $K$, the equation $z_1+\cdots +z_n=1$ is said to be \emph{non-degenerate} if no non-trivial subsum of the left-hand side is equal to zero; that is, whenever $I$ is a nonempty subset of $\{1,\ldots ,n\}$, we have $\sum_{i\in I} z_i\neq 0$. 

The $S$-unit theorem (see \cite[Theorem 6.1.3]{EG}) is a hugely significant result in Diophantine approximation, which we state for the reader's convenience. 

\begin{thm}\label{thm:unit}
  Let $K$ be a field of characteristic zero, let
  $a_1,\dots,a_n$ be nonzero elements of $K$, and let $H\subset(K^*)^n$ be a finitely generated multiplicative subgroup.
  Then there are only finitely many non-degenerate solutions $(x_1,\dots,x_n)\in H$ to the equation
  \begin{equation}
  a_1x_1+\cdots+a_nx_n=1.\label{eq:Sunit}
  \end{equation} 
  \end{thm}

We will use a quantitative version of the $S$-unit theorem. There are a number of quantitative versions (see for example \cite{AV, EG, ESS02, Sch}), but we find the following version, due to Amoroso and Viada \cite[Theorem 6.2]{AV}, most convenient for our purposes.  We note that Amoroso and Viada assume their fields are algebraically closed throughout, but for the statement given below this hypothesis is unnecessary since we can embed a field into its algebraic closure.  We recall that a finitely generated abelian group is isomorphic to the direct sum of a finite group along with a group isomorphic to $\mathbb{Z}^r$ for some $r\ge 0$; the quantity $r$ is uniquely determined by the group and is called the \emph{rank} of the group.  

\begin{thm} \label{thm:av} 
Let $K$ be a field of characteristic zero, let $a_1, \dots, a_n $ be nonzero elements of $K$, and let $\Gamma$ be a finitely generated multiplicative subgroup of $(K^*)^n$ of rank $r < \infty$. Then there are at most $$(8n)^{4n^4(n+r+1)}$$ non-degenerate solutions to the equation $$a_1 x_1+\cdots + a_n x_n =1$$  with $(x_1, \dots, x_n) \in \Gamma$.
\end{thm}
We note that all versions of the $S$-unit theorem are ineffective, except in the case when $n\le 2$.  

\section{Proof of Theorem \ref{thm:main1}} \label{cobham}

In this section, we prove the following version of Theorem \ref{thm:main1}.  

\begin{thm} \label{thm:main2}
Let $k$ and $\ell$ be multiplicatively independent positive integers, let $d\ge 2$, and let 
$\Gamma= (Q, \Sigma_k^d, \delta, q_0, F)$ and $\Gamma' = (Q', \Sigma_{\ell}^d, \delta', q_0', F')$ be deterministic finite-state automata accepting sparse regular languages $\mathcal{L}\subseteq (\Sigma_k^d)^*$ and $\mathcal{L}'\subseteq (\Sigma_{\ell}^d)^*$. If $X\subseteq \mathbb{N}^d$ is the set of $d$-tuples of natural numbers whose base-$k$ expansions are elements of $\mathcal{L}$ and $Y\subseteq \mathbb{N}^d$ is the set of $d$-tuples of natural numbers whose base-$\ell$ expansions are elements of $\mathcal{L}'$, then 
\begin{equation}
\label{eq:intersectionbound}
|X\cap Y|\le 
k^{d|Q|}\cdot \ell^{d|Q'|}\cdot \left(8(|Q|+|Q'|-1)\right)^{10d(|Q|+|Q'|)^5}.\end{equation}

\end{thm}

We note that we have not attempted to optimize the upper bounds, as to do so would lead to unwieldy expressions.  Nevertheless, the bounds we obtain cannot be significantly improved using our methods. We begin with a basic estimate. 

\begin{prop}\label{prop:sets}
Let $N\ge 2$, let $\Sigma$ be a finite alphabet of size $N$, 
and let $\Gamma= (Q, \Sigma, \delta, q_0, F)$ be a 
deterministic finite automaton accepting a sparse language $\mathcal{L}$. Then $\mathcal{L}$ is a 
finite (possibly empty) union of at most $$(|Q|-1)!(N^{|
Q|-1}+N^{|Q|-2}+\cdots +1)$$ languages of the form
$$\{v_0 w_1^* v_1 w_2^* \cdots v_{s-1} w_s^* v_{s}\}$$ with $w_1,\ldots ,w_s, v_1,\ldots , v_{s}$ words in $\Sigma^*$ in which the $w_i$ are non-empty but the $v_i$ may be empty and with $|w_1|+\cdots +|w_s|\le |Q|-1$ and $|v_0|+\cdots +|v_{s}|\le N(|Q|-1)$.  
\end{prop}

\begin{proof} 
Suppose towards a contradiction that this is not the case and pick a DFA $(Q, \Sigma, \delta, q_0,F)$ for which the conclusion to the statement of the proposition does not hold with $|Q|$ minimal.  We note that the result holds when $|Q|=1$, as the only sparse set accepted by a one-state automaton with input alphabet of size at least two is the empty set.  Thus we may assume that $|Q|>1$.  

We put a transitive binary relation $\preceq$ on $Q$ by declaring that $q\preceq q'$ for $q,q'\in Q$ if there is a word $w\in \Sigma^*$ such that $\delta(q,w)=q'$.  We then declare that two states $q,q'$ are equivalent if $q\preceq q'$ and $q'\preceq q$.  Then this relation is reflexive as $\delta(q,\epsilon)=q$, where $\epsilon$ is the empty word; and it is symmetric and transitive by construction. We let $[q]$ denote the equivalence class of $q$.  Then $\preceq$ induces a partial order on the equivalence classes. We let $r$ denote the size of the equivalence class $[q_0]$. Since $\mathcal{L}$ is non-empty, there is at least one path from $q_0$ to an accepting state.  In particular, by Proposition \ref{sparse} (3), we have that there is at most one cycle based at $q_0$ and since it passes through all states in $[q_0]$, this cycle, if it exists, is some word $w_1$ of length $r$.  We note that if $r\ge 2$ then there must be a cycle based at $q_0$, but if $r=1$ it is possible that $\delta(q_0,w)=q_0$ if and only if $w$ is the empty word.  
We now consider two cases corresponding to these possibilities.  The simpler case is when $\delta(q_0,w)=q_0$ only if $w$ is the empty word.  In this case, $[q_0]=q_0$ and for each $x\in \Sigma$, we let $\mathcal{L}_x$ denote the set of all words $w\in \Sigma^*$ whose first letter is $x$ and for which $w\in \mathcal{L}$.  Then $\delta(q_0, u_x)\in Q\setminus \{q_0\}$ for every $u_x\in \mathcal{L}_x$.  Then since $q_0$ is only equivalent to itself, we see that $\mathcal{L}_x= x\mathcal{E}_x$, where $\mathcal{E}_x$ is the regular language accepted by the automaton
$\Gamma_x:=(Q\setminus \{q_0\}, \Sigma, \delta, \delta(q_0,x),F\setminus \{q_0\})$.  

\noindent Then by minimality of $|Q|$, we have that $\mathcal{E}_x$ is a union of at most $$(|Q|-2)!(N^{|Q|-2}+\cdots +1)$$ sets of the form
$$\{v_0 w_1^* v_1 w_2^* \cdots v_{s-1} w_s^* v_{s}\}$$
with $w_1,\ldots ,w_s, v_0,\ldots , v_{s}$ words in $\Sigma^*$ in which the $w_i$ are non-empty but the $v_i$ may be empty and with $|w_1|+\cdots +|w_s|\le |Q|-2$ and $|v_0| + \cdots +|v_{s}| \le N(|Q|-2)$. 
Then $\mathcal{L}_x$ is a union of at most $(|Q|-2)! (N^{|Q|-2}+ \cdots+1)$ sets of the form
$$\{(xv_0) w_1^* \cdots v_{s-1} w_s^* v_{s}\}$$
with $w_1,\ldots ,w_s, v_0,\ldots , v_{s}$ words in $\Sigma^*$  in which the $w_i$ are non-empty but the $v_i$ may be empty and with $|w_1|+\cdots +|w_s|\le |Q|-1$ and $|xv_0|+\cdots +|v_{s}|\le N(|Q|-1)$, since $N\ge 1$. 
Then since $\mathcal{L}$ is the union of $\mathcal{L}_x$ for $x\in \Sigma$ we see that $\mathcal{L}$ is a union of at most $(|Q|-2)!(N^{|Q|-1}+N^{|Q|-2}+\cdots +N)$ sets of the form 
$$\{v_0 w_1^* v_1 w_2^* \cdots v_{s-1} w_s^* v_{s}\}$$ with $w_1,\ldots ,w_s, v_0,\ldots , v_{s}$ words in $\Sigma^*$ in which the $w_i$ are non-empty but the $v_i$ may be empty and with $|w_1|+\cdots +|w_s|\le |Q|-1$ and $|v_0|+\cdots +|v_{s}|\le N(|Q|-1)$.  Thus we obtain the result in this case.

We next consider the case when there is a unique cycle $w_1$ of length $r\ge 1$ based at $q_0$.  In particular, $ [q_0] $ has size $r$.
Then in this case we can write $\mathcal{L}=\mathcal{L}_0\cup \mathcal{L}_1$, where $\mathcal{L}_0$ is the set of words $w$ in $\mathcal{L}$ for which $\delta(q_0,w)\in [q_0]$ and $\mathcal{L}_1$ is the set of words $w\in \mathcal{L}$ for which $\delta(q_0,w)\not\in [q_0]$.  By construction every word in $\mathcal{L}_0$ is of the form $w_1^* v$ where $v$ is a proper prefix of $w_1$.  In particular, $\mathcal{L}_0$ is a union of at most $r$ sets of the desired form. We next consider $\mathcal{L}_1$.
If $w\in \mathcal{L}_1$ then $w$ can be written as $uxv$ with $u,v \in \Sigma^*, x \in \Sigma$ such that $\delta(q_0,u)\in [q_0]$ but $\delta(q_0,ux)\not\in [q_0]$.  Then we may write $\mathcal{L}_1$ as a union of $|Q|-r$ sublanguages
$\mathcal{L}_{1,q}$ for each $q\in Q\setminus [q_0]$, where $\mathcal{L}_1$ is the set of words in $\mathcal{L}$ of the form $uxv$ with $\delta(q_0,u)\in [q_0]$, $\delta(q_0,ux)=q$. 
 
Then each $\mathcal{L}_{1,q}$ is a finite union of languages of the form $w_1^* z x \mathcal{E}_q$ where $z$ is a proper prefix of $w_1$, $x\in \Sigma$ and $\delta(q_0,z x)=q$ and $\mathcal{E}_q$ is a sparse language accepted by an automaton with state set $Q \setminus [q_0]$.  
In particular, by minimality of $|Q|$, each $\mathcal{E}_q$ is a finite union of at most $(|Q|-r-1)!(N^{|Q|-r-1}+\cdots+1)$ sets of the form
$$\{v_1 w_2^* v_2 w_3^* \cdots v_{s-1} w_s^* v_{s}\}$$ with $|w_2|+\cdots + |w_s|\le |Q|-r-1$ and $|v_1|+\cdots +|v_{s}|\le N(|Q|-r-1)$. Then since $w$ has at most $r$ proper prefixes and since there are at most $N$ choices for $x$, we see that $\mathcal{L}_{1,q}$ is a union of at most $(rN)(|Q|-r-1)! (N^{|Q|-r-1}+\cdots + N+1)$ sets of the form
  $$\{w_1^* (zxv_1) w_2^* v_2 w_3^* \cdots v_{s-1} w_s^* v_{s}\}$$ with
  $|w_1|+\cdots +|w_s|\le |Q|-1$ and $|zxv_1|+\cdots +|v_{s}|\le N(|Q|-r-1)+r \le N(|Q|-1)$.
  Thus $\mathcal{L}$ is a union of at most $(|Q|-r)rN (|Q|-r-1)! (N^{|Q|-r-1}+\cdots + N+1) +r$ sets of the desired form, where the contribution of $r$ comes from considering our decomposition of $\mathcal{L}_0$ and the $|Q|-r$ factor comes from considering the languages $\mathcal{L}_{1,q}$ for $q\in Q\setminus [q_0]$. Finally, since $N\ge 2$ we have $r\le N^{r-1}$ and so
    \begin{align*}
  &~ (|Q|-r)rN(|Q|-r-1)! (N^{|Q|-r-1}+\cdots + N+1) +r\\
  & \le (|Q|-r)N^{r-1}\cdot N (|Q|-r-1)! (N^{|Q|-r-1}+\cdots + N+1) + N^{r-1}\\
  &= (|Q|-r)! (N^{|Q|-1}+\cdots + N^{r}) + N^{r-1} \\
  &\le (|Q|-1)!(N^{|Q|-1}+\cdots + 1).
  \end{align*}
  The result follows.
\end{proof}

We now make use of the estimate in Theorem \ref{thm:av}.

\begin{lemma} Let $k$ and $\ell$ be multiplicatively independent integers, let $m$ and $n$ be positive integers, and let $a_1,\ldots, a_n, b_{1},\ldots, b_{m}$ be nonzero rational numbers.  Then there are at most $$\left(8(n+m-1)\right)^{10(n+m)^5-4(n+m-1)^4}$$ to the equation
$$a_1X_1+\cdots +a_n X_n + b_{1} Y_{1}+\cdots + b_{m} Y_{m}=0$$ in which each $X_i$ is a power of $k$, each $Y_i$ is a power of $\ell$ and no proper non-trivial subsum of the left-hand side vanishes. 
\label{lem:ABC}
\end{lemma}
\begin{proof} We consider the case when $n\le m$; the case when $m<n$ is handled similarly. Let $H_1:=\{k^a \ell^b : a,b \in \mathbb{Z}\}\cong (\mathbb{Z},+)^2$, which is an abelian group of rank $2$, and let $H_2:=\{\ell^b : b \in \mathbb{Z}\}\cong (\mathbb{Z},+)$, which has rank one.

A solution to the equation 
$$a_1X_1+\cdots +a_{n} X_n + b_{1} Y_{1}+\cdots + b_{m} Y_{m}=0$$ with the desired properties gives rise to a solution to the equation \begin{equation} \label{eq:Zi}
\sum_{i=1}^n (-a_i/b_m)Z_1 + \sum_{j=1}^{m-1} (-b_j/b_m) Z_{n+j}  = 1
\end{equation}
with $Z_i=X_i/Y_m\in H_1$ for $1 \le i\le n$ and $Z_i=Y_{i-n}/Y_m\in H_2$ for $n+1\le i<n+m$, and with Equation (\ref{eq:Zi}) non-degenerate. Then $(Z_1,\ldots ,Z_{n+m-1})\in \Gamma\subseteq (\mathbb{Q}^*)^{n+m-1}$, where 
$\Gamma = H_1^{n}\times H_2^{m-1}$, which is a group of rank $2n+m-1$.
Thus we can take $r=2n+m-1$ in Theorem \ref{thm:av}, and this gives that Equation (\ref{eq:Zi}) has at most 
\begin{equation} \label{eq:sol}
(8(n+m-1))^{4(n+m-1)^4 (3n+2m-1)}\end{equation} non-degenerate solutions. Since $n\le m$, we have $3n+2m \le 5(n+m)/2$, and so
$4(n+m-1)^4 (3n+2m-1)\le 10(n+m)^5 - 4(n+m-1)^4$.
Thus the quantity in Equation (\ref{eq:sol}) is bounded above by $$(8(n+m-1))^{10(n+m)^5-4(n+m-1)^4}.$$ 

Finally, we can uniquely recover the original $X_i$'s and $Y_i$'s from $Z_1,\ldots ,Z_{m+n-1}$. To see this, observe that for $i=1,\ldots ,n$ we must have $Z_i=k^a/\ell^b$ for some integers $a$ and $b$. Since $k$ and $\ell$ are multiplicatively independent, $a$ and $b$ are uniquely determined. So we can recover $X_1,\ldots ,X_n$ and $Y_m$ from $Z_1, \dots, Z_n$. But we can then recover $Y_1, \dots, Y_{m-1}$ from the remaining $Z_j$.  The result follows.

\end{proof}

We now use the preceding lemma to give estimates in the case where some degeneracy is allowed.

\begin{lemma} Let $k$ and $\ell$ be multiplicatively independent integers and let $m, n\ge 1$ be integers and let $a_1,\ldots, a_n, b_{1},\ldots ,b_{m}$ be nonzero rational numbers.  Then there are at most $$
2^{-(n+m)}\cdot \left( 8(n+m-1) \right)^{10(n+m)^5-(n+m)}$$ solutions to the equation
$$a_1X_1+\cdots +a_n X_n + b_{1} Y_{1}+\cdots + b_{m} Y_{m}=0$$ in which each $X_i$ is a power of $k$, each $Y_i$ is a power of $\ell$, and no non-trivial subsum of either
$a_1X_1+\cdots +a_n X_n$ or $b_{1} Y_{1}+\cdots + b_{m} Y_{m}$ vanishes.
\label{lem:upper}
\end{lemma}

\begin{proof}
For each solution to
$$a_1X_1+\cdots +a_n X_n + b_1 Y_1+\cdots + b_m Y_m=0$$ such that no subsum of either $a_1X_1+\cdots +a_n X_n$ or $b_1 Y_1+\cdots + b_m Y_m$ vanishes, we can associate a set partition 
$\pi$ of the set $V:=\{X_1,\ldots, X_n, Y_1,\ldots ,Y_m\}$ into disjoint non-empty subsets $U_1,\ldots ,U_r$ such that the subsum corresponding to the variables in each $U_i$ vanishes and no proper subsum vanishes. Let $c_i := |U_i|$. Then $U_i$ intersects both $\{X_1,\ldots ,X_n\}$ and $\{Y_1,\ldots ,Y_m\}$ non-trivially and so by Lemma \ref{lem:ABC}, for $i=1,\ldots ,r$, there are at most $\left(8(c_i-1)\right)^{10c_i^5-4(c_i-1)^4}$ non-degenerate solutions to the subsum 
$$\sum_{X_j\in U_i} a_j X_j + \sum_{Y_j\in U_i} b_j Y_j = 0$$ with each $X_j$ a power of $k$ and each $Y_j$ a power of $\ell$.   Thus for the set partition $\pi$ we have at most 
\begin{align*}
\prod_{i=1}^r  \left(8(c_i-1)\right)^{10c_i^5-4(c_i-1)^4} &\le \left(8(n+m-1)\right)^{\sum_{i=1}^r (10c_i^5 - 4(c_i-1)^4)} \\
&\le \left(8(n+m-1)\right)^{10\left(\sum_{i=1}^r c_i\right)^5 - 4\sum_{i=1}^r (c_i-1)}\\
&=  \left(8(n+m-1)\right)^{10(n+m)^5 - 4(n+m-r)}\\
&\le \left(8(n+m-1)\right)^{10(n+m)^5 - 2(n+m)}
\end{align*}
solutions, where the last step follows from the fact that
\begin{equation}\label{eq:r}
r\le (n+m)/2,\end{equation}
which is a consequence of the fact that each $U_i$ intersects both $\{X_1,\ldots , X_n\}$ and $\{Y_1,\ldots ,Y_m\}$ non-trivially. 

Finally, observe that collection of set partitions of a finite set $W$ having exactly $e$ parts embeds in the collection of surjective maps from $W$ to $\{1,\ldots ,e\}$, by first assigning the labels $1,\ldots ,e$ to the sets making up a set partition and then associating the map which sends $w\in W$ to the label of the set it is in.  Since the number of parts in our set partitions is bounded by $(n+m)/2$, we see that the number of possible set partitions we have to consider is at most $\left( (n+m)/2 \right)^{n+m}$, since it embeds in the set of maps from $V$ into $\{1,\ldots ,\lfloor (n+m)/2\rfloor\}$.
Since we get at most 
$$\left(8(n+m-1)\right)^{ 10(n+m)^5-2(n+m)}$$ solutions of the desired form corresponding to each associated set partition of $V$, and since there are at most $((n+m)/2)^{n+m}$ possible set partitions that can occur, we get an upper bound of
\begin{align*} &~ ((n+m)/2)^{n+m} \cdot \left(8(n+m-1)\right)^{ 10(n+m)^5-2(n+m)} \\
&\le 2^{-(n+m)}\cdot \left(8(n+m-1)\right)^{n+m} (8(n+m-1))^{10(n+m)^5-2(n+m)}\\
&=  2^{-(n+m)}\cdot  (8(n+m-1))^{10(n+m)^5-(n+m)}.
\end{align*} The result follows.
\end{proof}

We now use these estimates to obtain upper bounds on the size of an intersection of sparse subsets of $\mathbb{N}$.

\begin{prop} \label{sets} Let $k$ and $\ell$ be multiplicatively independent positive integers, let $s\ge 1, t\ge 1$, and let $v_0,\ldots ,v_{s},w_1,\ldots ,w_s\in \Sigma_k^*$ and $u_0,\ldots ,u_{t}, y_1,\ldots ,y_t\in \Sigma_{\ell}^*$.  
If 
$$X=\{[v_0 w_1^* v_1 w_2^* \cdots v_{s-1} w_s^* v_{s}]_k\}$$ and 
$$Y=\{[u_0 y_1^* u_1 y_2^* \cdots u_{t-1} y_t^* u_{t}]_{\ell}\},$$ 
 then $$\left| X\cap Y\right| \le  \left(8(s+t+1)\right)^{10(s+t+2)^5-(s+t+2)}.$$
\end{prop}
\begin{proof}
By Proposition \ref{rem:sparse} we have that $X$ is of the form
 $$\left\{c_0 + c_1 k^{\delta_s n_s} + c_2 k^{\delta_s n_s + \delta_{s-1} n_{s-1}}+\cdots + c_s k^{\delta_s n_s+\cdots +\delta_1 n_1}  \colon n_1,\ldots ,n_s\ge 0\right\},$$ where $c_0,\ldots ,c_s$ are rational numbers. 
 Similarly, $Y$ is of the form
  $$\left\{ d_0 + d_1 \ell^{\delta_t' m_t} + d_2 \ell^{\delta_t' m_s + \delta_{t-1}' m_{t-1}}+\cdots + d_t \ell^{\delta_t'm_t+\cdots +\delta_1' m_1}\colon m_1,\ldots, m_t\ge 0\right\},$$ where $d_0,\ldots, d_t$ are rational numbers.
  
Then an element in $X\cap Y$ corresponds to a solution to the equation 
$$d_0 X_0+\cdots +d_t X_t -c_0 Y_0 - \cdots -c_s Y_s=0,$$
where $$X_0=1, X_1= \ell^{\delta_t' m_t},\ldots , X_t= \ell^{\delta_t' m_t+\cdots +\delta_1' m_1}$$ and
$$Y_{0}=1, \ldots , Y_{s}= k^{\delta_s n_s+\cdots +\delta_1 n_1},$$ with the corresponding element in the intersection given by $$A:=d_0 X_0+\cdots + d_t X_t = c_0 Y_0 + \cdots + c_s Y_s.$$

Since we are only concerned about the quantity $A$ in determining $X\cap Y$, after removing a maximal vanishing subsum\footnote{For nonzero $A$, this will be necessarily a proper subset, but when $A=0$ this will be the entire set.  The estimates we give account for this possibility.} we may assume that no non-trivial subsum of the terms involving powers of $\ell$ vanishes and that there are at most $t+1$ such terms. Similarly, we may remove a maximal vanishing subsum from $X_{t+1}+\cdots +X_{t+s+1}$. 

By Lemma \ref{lem:upper}, taking $n$ to be the number of terms from our first sum, we have $n\le t+1$; similarly, we can take $m$ to be the number of terms from our second subsum and we have $m\le s+1$. Using the fact that there are at most $2^{s+1}\cdot 2^{t+1}$ possible pairs of maximal vanishing subsums that we can remove and the fact that the function $$F(a,b)=2^{-a-b-2}\left(8(a+b+1)\right)^{10(a+b+2)^5-(a+b+2)}$$ is increasing in both $a$ and $b$ for $a,b\ge 0$, we then see there are at most
$$2^{s+t+2}\cdot 2^{-(s+t+2)}\left(8(s+t+1)\right)^{10(s+t+2)^5-(s+t+2)}$$ elements in $X\cap Y$.  The result follows.

\end{proof}

We are now ready to prove Theorem \ref{thm:main2}.
\begin{proof}[Proof of Theorem \ref{thm:main2}]
By Proposition \ref{prop:sets} $X$ is a union of sets $W_1,\ldots ,W_{A_1}$ of the form 
\begin{equation}
\label{eq:form1}
\{[v_0 w_1^* v_1 w_2^* \cdots v_{s-1} w_s^* v_{s}]_k\}
\end{equation}
 with $w_1,\ldots ,w_s, v_0,\ldots , v_{s}$ words in $(\Sigma_k^d)^*$ in which the $w_i$ are non-empty but the $v_i$ may be empty and with $|w_1|+\cdots +|w_s|\le |Q|-1$ and $|v_0|+\cdots +|v_{s}|\le k^d(|Q|-1)$.  Moreover, since our input alphabet has size $k^d$, Proposition \ref{prop:sets} also says we can take
$$A_1\le (|Q|-1)!(k^{d(|Q|-1)}+k^{d(|Q|-2)}+\cdots +1).$$ Similarly,
$Y$ is the union of sets $Z_{1},\ldots, Z_{A_2}$ with $$A_2\le (|Q'|-1)!(\ell^{d(|Q'|-1)}+\ell^{d(|Q'|-2)}+\cdots +1)$$ and each $Z_j$ of the form
\begin{equation}
\label{eq:form2}
\{[u_0 y_1^* u_1 y_2^* \cdots u_{t-1} y_t^* u_{t+1}]_{\ell}\}
\end{equation}
 with $u_0,\ldots ,u_{t+1}, y_1,\ldots , y_{t}$ words in $(\Sigma_{\ell}^d)^*$ in which the $y_i$ are non-empty but the $u_i$ may be empty and with $|y_1|+\cdots +|y_t|\le |Q'|-1$ and $|u_0|+\cdots +|u_{t}|\le \ell (|Q'|-1)$.  In particular, $t\le |Q'|-1$ for each such set.
 
Now for $(p,q)\in \{1,\ldots ,A_1\}\times \{1,\ldots ,A_2\}$ and each $i=1,\ldots ,d$, we let $W_{p,i}\subseteq \mathbb{N}$ and $Z_{q,i}\subseteq \mathbb{N}$ be respectively the images of $W_p$ and $Z_q$ under the projection map from $\mathbb{N}^d$ onto its $i$-th coordinate.  For $i\in \{1,\ldots ,d\}$ we have 
$$\proj_i(v_1 w_1^* v_2 w_2^* \cdots v_s w_s^* v_{s+1})= \proj_i(v_1) \proj_i(w_1)^* \cdots \proj_i(w_s)^*\proj_i(v_{s+1}),$$ where ${\rm proj}_i$ is the projection map $(\Sigma_k^d)^*\to \Sigma_k^*$ obtained by taking the $i$-th coordinate. It follows that each $W_{p,i}$ is a set of the form given in Equation (\ref{eq:form1}) but where we now use words over $\Sigma_k$ instead of $(\Sigma_k)^d$.  Similarly, each $Z_{q,i}$ is a set of the form given in Equation (\ref{eq:form2}), but where we now use words over $\Sigma_{\ell}$ instead. 


By Proposition \ref{sets}, each $W_{p,i}\cap Z_{q,i}$ has cardinality at most 
$$\left(8(s+t+1)\right)^{10(s+t+2)^5-(s+t+2)}$$ in the intersection.  In particular, since $s\le |Q|-1$ and $t\le |Q'|-1$, 
we have
$$|W_{p,i}\cap Z_{q,i}| \le \left(8(|Q|+|Q'|-1)\right)^{10(|Q|+|Q'|)^5-(|Q|+|Q'|)}.$$
Now since $$W_p\cap Z_q \subseteq (W_{p,1}\cap Z_{q,1})\times \cdots  \times (W_{p,d}\cap Z_{q,d}),$$ we then see each intersection 
$W_p\cap Z_q$ has size at most
$$ \left(8(|Q|+|Q'|-1)\right)^{10d(|Q|+|Q'|)^5-d(|Q|+|Q'|)}.$$
Finally, since
$$X\cap Y= \bigcup_{p\le A_1}\bigcup_{q\le A_2} (W_p\cap Z_q),$$ we see that 
$$|X\cap Y| \le A_1\cdot A_2 \cdot \left(8(|Q|+|Q'|-1)\right)^{10d(|Q|+|Q'|)^5-d(|Q|+|Q'|)}.$$
Finally, observe that 
$A_1\le |Q|^{|Q|} \cdot k^{d|Q|}$ and $A_2\le |Q'|^{|Q'|} \cdot \ell^{d|Q'|}$, so we get 
$$|X\cap Y| \le k^{d|Q|} \cdot \ell^{d|Q'|} \cdot |Q|^{|Q|} \cdot |Q'|^{|Q'|} \cdot \left(8(|Q|+|Q'|-1)\right)^{10d(|Q|+|Q'|)^5-d(|Q|+|Q'|)},$$ which is easily seen to be less than
$$ k^{d|Q|} \cdot \ell^{d|Q'|} \cdot \left(8(|Q|+|Q'|-1)\right)^{10d(|Q|+|Q'|)^5}.$$
The result follows.
\end{proof}

\begin{remark} We note that the strategy employed in the proof of Theorem \ref{thm:main2} involves giving a description of the complexity of the sublanguages of $(\Sigma_k^d)^*$ and $(\Sigma_{\ell}^d)^*$ accepted by our automata, then using this to bound the complexity of their projections, and finally using $S$-unit theory to get a bound on the sizes of the projections.  An alternative approach would be to first find the automata that accept the projections of the languages and work with those bounds.  The projection of a regular language accepted by an automaton with $n$ states can be accepted by an automaton with $2^n$ states.\footnote{There are improvements to this bound (see, for example, \cite{J} and references therein), but in general the number of states required to accept a projected language is exponential in the number of states of the minimal automaton accepting the original language.}  If one uses this approach one gets an alternative bound that is typically much worse.
\end{remark}
\section{A general intersection question} \label{conjecture}

We now consider the general question of what the intersection of a sparse automatic set with a zero-density automatic set can look like.  We recall that for a subset 
$S$ of $\mathbb{N}$, the \emph{density} of $S$ is just the limit
\begin{equation}
\lim_{n \to \infty} \frac{\pi_S(n)}{n},
\end{equation}
if it exists.  In general, a set of natural numbers always has a lower density and an upper density
given respectively by
\begin{equation}
\liminf_{n \to \infty}  \frac{\pi_S(n)}{n} \text{ and } \limsup_{n \to \infty}  \frac{\pi_S(n)}{n},
\end{equation} and so the density exists precisely when these two values coincide.

We make the remark that since sparse automatic sets are polylogarithmically bounded, they necessarily have density zero. 

The following result is due to the second-named author
\cite[Prop. 2.1]{bellaut}.

\begin{prop} Let $k\ge 2$ be a natural number, let $h:\mathbb{N}\to \mathbb{Q}_{\ge 0}$ be a $k$-automatic sequence, and let $s(n)=\sum_{j<n} h(j)$.  Then there exist $\beta\in (0,k)$, $C>0$, $a\ge 1$, and nonnegative rational numbers $c_{j}$ for $j\in  \{0,1,\ldots ,a-1\}$ such that 
$$|s(k^{an+j}) - c_{j} k^{an+j}|<C \beta^{an}$$ for every $n\ge 0$.  Moreover, $a$ and the rational numbers $c_{0},\ldots , c_{a-1}$ are recursively computable and $\beta$ can be effectively determined. 
\label{thm:perron}
\end{prop}

As a consequence of this, we can prove that either a $k$-automatic set $S$ has positive lower density (i.e., $\liminf \pi_S(x)/x > 0$) or there is some positive $\epsilon>0$ such that
$\pi_S(x)={O}(x^{1-\epsilon})$.

\begin{prop} Let $k\ge 2$ be a natural number and let $S$ be a $k$-automatic subset of the natural numbers. Then either $S$ has positive lower density or there is some $\epsilon>0$ such that
$\pi_S(x)={O}(x^{1-\epsilon}).$
\end{prop}

\begin{proof}
Taking $h:\mathbb{N}\to \{0,1\}$ to be the characteristic function of $S$ and then applying Proposition \ref{thm:perron}, we see that either $S$ has positive lower density or $\pi_S(k^n) = {O}(\beta^n)$ for some $\beta\in (0,k)$. We henceforth assume that we are in the second case.  Then there is some $\epsilon>0$ such that 
$$\pi_S(k^n) = {O}(k^{(1-\epsilon)n}).$$ 
Then for a given $x>1$, we have $k^n\le x<k^{n+1}$ for some $n$ and so
$$\pi_S(x) \le \pi_S(k^{n+1}) = {O}((k^{n+1})^{1-\epsilon}).$$ Since $kx\ge k^{n+1}$ we then see 
$$\pi_S(x)={O}(x^{1-\epsilon}),$$ and so we obtain the desired result.
\end{proof}
In general, if $k$ and $\ell$ are multiplicatively independent, then a sparse $k$-automatic set can have infinite intersection with an $\ell$-automatic set, but in the case when $X$ is a sparse $k$-automatic set and $Y$ is an $\ell$-automatic set of zero density, we expect $X\cap Y$ to be finite.  Heuristically, one can see why this should be the case as follows. Since $Y$ has zero density, we have shown that there is some $\epsilon>0$ such that
$\pi_Y(x)=O(x^{1-\epsilon})$, and since $X$ is sparse there are positive constants $c$ and $d$ such that $\pi_X(x)\le c(\log\, x)^d$ for $x$ large. 
Thus there is some $C>0$ such that, for $x$ large, if we take a natural number in $[0,x]$, the probability that it lies in $Y$ is at most $Cx^{-\epsilon}$.
In particular, if $i_1<i_2<i_3<\cdots $ is an enumeration of the elements of our sparse $k$-automatic set $S$, then since the bases $k$ and $\ell$ are multiplicatively independent, we expect that the probability that $i_j$ is in $Y$ to be at most $Ci_j^{-\epsilon}$, and so the expected number of elements in $X\cap Y$ should be bounded by the size of the 
sum
$$\sum_{j\ge 1} \frac{C}{i_j^{\epsilon}}.$$
Notice that the above series converges when $X$ is sparse. To see this, recall that $\pi_X(x)\le c(\log\, x)^d$ for some $c,d>0$ and for $x$ large.  Since
$\pi_X(i_N)=N$, we then have
$N\le c(\log\, i_N)^d$ for $N$ large, which gives $i_N \ge \exp((N/c)^{1/d})$ for $N$ sufficiently large. In particular, $i_N$ grows faster than any polynomial in $N$ and so for every $\epsilon>0$ we have
that $\sum 1/i_j^{\epsilon}$ converges. 
 
Using this heuristic as a guide, we make the following conjecture, although this problem appears to be well beyond what current methods in number theory can handle.

\begin{conj}\label{conj:gen}
Let $k, \ell$ be multiplicatively independent positive integers. If $X$ is a sparse $k$-automatic subset of $\mathbb{N}$ and $Y$ is a zero-density $\ell$-automatic subset of $\mathbb{N}$, then $X \cap Y$ is finite. 
\end{conj}
We note that if we take $k=2,\ell=3$ and $X=\{2^i\colon i\ge 0\}$ and $Y$ to be the set of numbers whose ternary expansions have no occurrences of $2$, then $Y$ has zero density and $X$ is sparse and so the conjecture of 
Erd\H{o}s \cite[p. 67]{Erdos} mentioned in the introduction is a special case of Conjecture \ref{conj:gen}.


\end{document}